\documentclass[aps,prx,print,twocolumn,showpacs,longbibliography]{revtex4-1} %
\usepackage{amsthm}
\usepackage{amsmath}
\usepackage{amsfonts}
\usepackage{graphicx}
\usepackage{dcolumn}
\usepackage{longtable}
\usepackage{mathrsfs}
\usepackage{setspace}
\usepackage{multirow}
\usepackage{float}
\usepackage{verbatim}
\usepackage{algorithmicx}
\usepackage{algorithm}
\usepackage{algpseudocode}
\usepackage{float}
\usepackage{afterpage}
\usepackage{xfrac}
\usepackage{blindtext}
\usepackage{enumitem}
\usepackage{amssymb}
\usepackage{tabularx}
\newcolumntype{D}{>{\centering\arraybackslash}X}

\usepackage{hyperref}
\hypersetup{
    bookmarksnumbered=true, 
    unicode=false, 
    pdfstartview={FitH}, 
    pdftitle={}, 
    pdfauthor={}, 
    pdfsubject={}, 
    pdfcreator={}, 
    pdfproducer={}, 
    pdfkeywords={}, 
    pdfnewwindow=true, 
    colorlinks=true, 
    linkcolor=blue, 
    citecolor=blue, 
    filecolor=blue, 
    urlcolor=blue 
}

\usepackage{epstopdf}

\usepackage{bbold} 

\newcommand{\ket}[1]{|#1\rangle}

\newcommand{\Rep}{\operatorname{Re}}
\newcommand{\Imp}{\operatorname{Im}}
\newtheorem{theorem}{Theorem}
\newtheorem{definition}{Definition}

\newtheorem{thm}{Theorem}

\newtheorem{lemma}[theorem]{Lemma}

\newcommand{\tjyoder}[1]{{#1}}
\newcommand{\ghl}[1]{{#1}}

\begin{document}

\title{The methodology of resonant equiangular composite quantum gates}

\author{Guang Hao Low, Theodore J. Yoder, and Isaac L. Chuang}
\affiliation{Department of Physics and Research Laboratory of Electronics, Massachusetts Institute of Technology, Cambridge, Massachusetts 02139, USA}
\date{\today}

\begin{abstract}
The creation of composite quantum gates that implement quantum response functions $\hat{U}(\theta)$ dependent on some parameter of interest $\theta$ is often more of an art than a science. 
Through inspired design, a sequence of $L$ primitive gates also depending on $\theta$ can engineer a highly nontrivial $\hat{U}(\theta)$ that enables myriad precision metrology, spectroscopy, and control techniques. 
However, discovering new, useful examples of $\hat{U}(\theta)$ requires great intuition to perceive the possibilities, and often brute-force to find optimal implementations.
We present a systematic and efficient methodology for composite gate design of arbitrary length, where phase-controlled primitive gates all rotating by $\theta$ act on a single spin. We fully characterize the realizable family of $\hat{U}(\theta)$, provide an efficient algorithm that decomposes a choice of $\hat{U}(\theta)$ into its shortest sequence of gates, and show how to efficiently choose an achievable $\hat{U}(\theta)$ that for fixed $L$, is an optimal approximation to objective functions on its quadratures. A strong connection is forged with \emph{classical} discrete-time signal processing, allowing us to
swiftly construct, as examples, compensated gates with optimal bandwidth that implement arbitrary single spin rotations with sub-wavelength spatial selectivity.
\end{abstract}


\maketitle

\section{Introduction}
Composite quantum gates~\cite{Vandersypen2005,Freeman1998spin} are indispensable to many important quantum technologies, such as nuclear magnetic resonance, magnetic resonance imaging, quantum sensing, and quantum computation~\cite{Nielsen2004}.
Their versatility arises from cunningly chosen sequences of $L$ primitive quantum gates that produce an effective quantum gate $\hat{U}$ with a more desirable dependence on some parameters of interest $\theta$, such as drive amplitude or background magnetic fields. As a function of $\theta$, the quantum response function $\hat{U}(\theta)$ can be tailored to
amplify weak signals beyond the statistics of repetition, and suppress noise without measurement. Finding such useful composite gates is thus the subject of intense research.

Even in single spin systems, \ghl{the focus of this work,} extraordinary richness can be found in the possible forms of $\hat{U}(\theta)$.
For example: NMR spectroscopy, where minute chemical shifts $\theta$ are made clearer through $\hat{U}(\theta)$~\cite{Levitt1986,Freeman1998spin,khaneja2005};
 Heisenberg-limited quantum imaging, where $\hat{U}(\theta)$ is made sensitive to sub-wavelength position variations $\theta$~\cite{arai2015fourier} without aliasing~\cite{Haberle2013high,Low2015};
 sub-wavelength spatial addressing where arbitrary quantum gates $\hat{U}(\theta)$ with low crosstalk are applied on spin arrays~\cite{Torosov2011,Merrill2014,Low2015}; quantum phase estimation for atomic clocks~\cite{Vanier2005} or tomography~\cite{Kimmel2015robust}, where extremely small drifts $\theta$ are amplified by factor $L$ in the gradient of $\hat{U}(\theta)$; error-compensation, where fractional control errors $\theta$ are exponentially suppressed like $\hat{U}(\theta)=\hat{U}(0)+\mathcal{O}(\theta^{\text{poly}(L)})$~\cite{Wimperis1994,Cummins2003,Brown2004,Uhrig2007,wang2014composite,Jones2013,Low2014,Soare2014}; quantum algorithms such as amplitude amplification~\cite{brassard1998quantum,Yoder2014} where a computation $\hat{U}(\theta)$ proceeds with input $\theta$. The discovery of other \tjyoder{applications} would be expedited if a useful characterization of all achievable $\hat{U}(\theta)$ could be found.
 
However, the road to new results does not end with a choice of $\hat{U}(\theta)$. \ghl{Given some reasonable system-dependent quantum control~\cite{glaser2015training}, its realization as a composite gate must be found}. Only with rare exceptions~\cite{Uhrig2007,Torosov2011,Low2015} and great effort are optimal, arbitrary length examples found in closed-form. Thus celebrated techniques including gradient ascent algorithms~\cite{khaneja2005} and pseudospectral methods~\cite{Ruths2012,Ross2012} formulate this as a systematic optimization problem that can be solved by brute force, but unfortunately with an exponential worst-case runtime $\mathcal{O}(e^L)$ for finding optimally short $L$ approximations.  Finding efficient solutions to various control problems would expand the potential of long composite gates, for which the most sophisticated quantum response functions can be constructed.

A tantalizing similarity is seen in discrete-time signal processing~\cite{Oppenheim2010}. Optimal finite-impulse response filters~\cite{Harris1978} can be designed simply by choosing the lowest-degree $L$ polynomial that is the optimal approximation to a desired frequency response, from which an optimal and exact implementation is computed -- made possible by efficient algorithms for both steps.
It is recognized that composite gates implement a filter on physical parameters~\cite{Soare2014,Kabytayev2014}, and the use of polynomials in quantum response functions is well-known~\cite{Li2009,Li2011}. Unfortunately, quantum constraints can render computing these polynomials and their optimal implementation a hard problem. It would be a tremendous advance if efficient solutions to these problems could be found, and even more so if the countless results from the exalted history of \emph{classical} discrete-time signal processing were transferable to the \emph{quantum} realm.

\ghl{One noteworthy step in this direction is the Shinnar-LeRoux algorithm~\cite{Shinnar1989,Pauly1991} and its refinements~\cite{Ikonomidou2000,Lee2007SLR,Grissom2014SLR}, which have so far been restricted to the field of magnetic resonance imaging. There, $\theta$ represents the amplitude of background magnetic fields and manifests as an off-resonant rotation. Given otherwise perfect and arbitrary single-spin control, this approach enables the efficient design of $\hat{U}(\theta)$ by a connection to finite-impulse response filters. Unfortunately, extending the concept to situations with different controls and additional restrictions, particularly the case of on-resonant compensating pulse sequences, appears to have been difficult.} 

\ghl{We contribute to this aspiration with a methodology, similar to the Shinnar-LeRoux algorithm}, for designing composite gates built from $L$ \ghl{resonant} primitive gates acting on a single spin, all rotating by angle $\theta$, but not necessarily with the same phase. We rigorously characterize the quantum response functions $\hat{U}(\theta)$ achievable in this manner, prove how we can efficiently \emph{choose} its form with trigonometric polynomials of degree $L$, and prove how its optimal implementation with the shortest sequence of primitive gates can be efficiently computed. In the process, a connection is made with discrete-time signal processing that allows us to inherit some of its machinery. These powerful tools expedite composite gate design, which we demonstrate with three optimal important examples: (1) narrowband and broadband composite population inversion gates; (2) compensated NOT gates with optimal bandwidth; (3) spatially-selective arbitrary single spin rotations below the diffraction limit.

We elucidate in Sec.~\ref{Sec:Model} the controls available to our single qubit system and demonstrate how equiangular composite gates motivate in Sec.~\ref{Sec:Design} the intuitive concept of choosing polynomials to explicitly define the quantum response function $\hat{U}(\theta)$. This is made rigorous and tractable in Sec.~\ref{Sec:characterization} by a simple characterization of the space of achievable $\hat{U}(\theta)$, and showing in Sec.~\ref{Sec:Implementation} how an optimal implementation of any such $\hat{U}(\theta)$ can be efficiently computed. We then show in Sec.~\ref{Sec:Computation} how an achievable $\hat{U}(\theta)$ can be efficiently computed from a partial specification with polynomials that describe only the composite gate fidelity or transition probability response functions. This  enables in Sec.~\ref{Sec:Selection} the efficient design of achievable $\hat{U}(\theta)$ by inheriting from discrete-time signal processing existing polynomials and efficient polynomial design algorithms. Together, these provide a methodology outlined in Sec.~\ref{Sec:methodology} for the systematic and efficient design of composite quantum gates. Use of this methodology is demonstrated in Sec.~\ref{Sec:Example} with the creation \tjyoder{of} optimal bandwidth compensated gates in Sec.~\ref{Sec:OBGates} that provide an optimal solution in Sec.~\ref{Sec:SpatialGates} to the problem of implementing sub-wavelength spatially selective arbitrary quantum gates. Further directions are discussed in Sec.\ref{Sec:Conclusion}.

\section{Model}
\label{Sec:Model}
The unitary quantum response function $\hat{U}(\theta)$ describes how a quantum system evolves under the influence of some parameter $\theta$ of interest. We consider the generic system of a \ghl{resonantly} driven single spin, and present a construction for composite quantum gates that motivates a powerful approach for designing their implemented response functions.

A two-level system driven by time-dependent Rabi frequency $\Omega(t)$ and phase $\phi(t)$ is controlled by the Hamiltonian $\hat{H}_c(t) = \frac{\Omega(t)}{2}\hat{\sigma}_{\phi(t)}$, where ${\hat{\sigma}_{\phi}=\cos(\phi) \hat{\sigma}_x+ \sin(\phi)\hat{\sigma}_y}$ and $\sigma_{x,y,z}$ are Pauli matrices.  Taking $\phi(t)=\phi$ to be constant over time $\tau$, we generate the primitive rotation
\begin{align}
\label{eq:PrimitiveRotation}
\hat{R}_{\phi}(\theta)=e^{-i \frac{\theta}{2}\hat{\sigma}_{\phi}},\quad \theta=\int_0^\tau \Omega(t)dt.
\end{align}
Note that $\Omega(t)$ might only be partially under our control and thus contain an uncontrollable residual signal. Our parameter of interest is thus $\theta$, which captures both effects.

An equiangular composite gate of length $L$ is built from these primitive rotations, each with the same rotation amplitude $\theta$, but with varying phases $\vec\phi=(\phi_1,...,\phi_L)$. This produces a $\theta$-dependent unitary $\hat{U}(\theta)$, or quantum response function, of the form
\begin{align}
\label{U_Expansion}
\hat{U}(\theta)&=\hat{R}_{\phi_L}(\theta)\hat{R}_{\phi_{L-1}}(\theta)\dotsc\hat{R}_{\phi_{1}}(\theta), \\ \nonumber
&=\sum_{j=0}^{L}(-i)^j \sin^j{\left(\frac{\theta}{2}\right)}\cos^{L-j}{\left(\frac{\theta}{2}\right)}\hat\Phi_{L,j},
\end{align}
where $\hat\Phi_{L,j}=\hat \sigma_x^j(\Rep[\Phi_{L,j}]\hat{\mathbb{1}}+i\Imp[\Phi_{L,j}]\hat \sigma_z)$, and the phase sums $\Phi_{L,j}$ are defined through the recurrence~\cite{Low2014}
\begin{align}
\label{Eq:PhaseSumRecurrance}
\Phi_{k,j}&=\Phi_{k-1,j}+\Phi_{k-1,j-1}e^{i(-1)^{j+1}\phi_j},\\ \nonumber \Phi_{0,0}&=1,\quad \Phi_{0,j\neq0}=0,
\end{align}
performed over $j=0,1,...,k$, then $k=1,2,...,L$.

Now, we make the crucial observation that $\hat{U}(\theta)$ is polynomial of degree $L$ in $x\equiv\cos{(\theta/2)}$ and $y\equiv\sin{(\theta/2)}$ with a particularly elegant representation. Using the trigonometric relation $x^2+y^2=1$, $\hat{U}(\theta)$ has the form
\begin{align}
\label{Eq:U_Polynomial}
\hat{U}(\theta)=
\resizebox{.75\columnwidth}{!}{
$
\begin{cases}
A(x)\hat{\mathbb{1}}+i B(x)\hat\sigma_z+i C(y)\hat\sigma_x+i D(y)\hat\sigma_y , & L\text{ odd}, \\
A(x)\hat{\mathbb{1}}+i B(x)\hat\sigma_z+i x C(y)\hat\sigma_x+i x D(y)\hat\sigma_y,& L\text{ even}.
\end{cases}
$
}
\end{align}
where $A(x),B(x),C(y),D(y)$ are polynomials of degree at most $L$ with coefficients $a_k,b_k,c_k,d_k$ \tjyoder{($k=0,1,\dots,L$)} respectively. In the following, $A,B,C,D$ without arguments are understood to be functions of the $x,y$ seen in Eq.~\ref{Eq:U_Polynomial}. As the tuple $(A,B,C,D)$ is an equivalent representation of $\hat{U}(\theta)$, we refer to both interchangeably. In particular, \emph{achievable} tuples are those than can be realized by some composite gate of Eq.~\ref{U_Expansion}:
\begin{definition}[Achievable polynomial tuples]
\label{Def:achievable}
A tuple of polynomials $(A,B,C,D)$ is achievable if $\exists L,\;\vec\phi\in\mathbb{R}^L$ s.t. $\hat{U}(\theta)=\hat{R}_{\phi_L}(\theta)\hat{R}_{\phi_{L-1}}(\theta)\dotsc\hat{R}_{\phi_{1}}$ has the form of Eq.~\ref{Eq:U_Polynomial}.
\end{definition}

We are often interested in only a few components of $(A,B,C,D)$. For example, the partial tuple $(A,\cdot,C,\cdot)$ fully defines the
gate fidelity response function $F_\chi(\theta)=\frac{1}{4}|\text{Tr}[\hat{R}^{\dagger}_0(\chi)\hat{U}]|^2$
with respect to some target gate $\hat{R}_0(\chi)$.
\begin{align}
\label{Eq:F_Polynomial}
F_\chi(\theta)=
\begin{cases}
\left|\cos{\left(\frac{\chi}{2}\right)}A-\sin{\left(\frac{\chi}{2}\right)} C\right|^2 , & L\text{ odd}, \\
\left|\cos{\left(\frac{\chi}{2}\right)}A-x\sin{\left(\frac{\chi}{2}\right)} C\right|^2,& L\text{ even}.
\end{cases}
\end{align}
Similarly, $(A,B,\cdot,\cdot)$ or $(\cdot,\cdot,C,D)$ fully defines the transition probability response function $p(\theta)=|\langle 0|\hat{U}|1\rangle|^2$,
\begin{align}
\label{Eq:P_Polynomial}
p(\theta)=1-A^2-B^2=(C^2+D^2)
\scriptstyle
\begin{cases}
1 , & L\text{ odd}, \\
x^2,& L\text{ even}.
\end{cases}
\end{align}
We refer to a tuple with $n$ empty slots as an $n$-partial tuple. An $n$-partial tuple is achievable if it is consistent with some achievable tuple.

A brute-force approach to composite gate design is minimizing an objective function for $\hat{U}(\theta)$ over a space $L\in\mathbb N, \vec{\phi}\in\mathbb{R}^L$. Though useful examples have been discovered in this manner, such an approach is highly unappealing. In addition to being inefficient with a runtime $\mathcal{O}(e^L)$, there is no guarantee that a globally optimal solution will be found. Furthermore, the procedure provides little of the necessary insight into possible $\hat{U}(\theta)$ for envisioning further novel applications.

\section{Systematic and efficient design of optimal composite gates}
\label{Sec:Design}
The functional form of $\hat{U}(\theta)$ hints at a powerful methodology for composite gate design via \emph{choices} of the polynomials $(A,B,C,D)$ of degree $L$. This ambition must solve long-standing problems: 
\begin{description}[align=left,noitemsep]
\item [(P1)] An insightful characterization of achievable $(A,B,C,D)$ to eliminate the traditional guesswork in envisioning novel quantum response functions and their dependence on $\vec{\phi}$.
\item [(P2)] An efficient algorithm to compute the optimal $\vec{\phi}$ implementing an achievable $(A,B,C,D)$, in contrast to the intractable  random search in time $\mathcal{O}(e^L)$ of current state-of-art~\cite{Low2015}. 
\item [(P3)] An efficient algorithm to compute an achievable $(A,B,C,D)$ from achievable partial tuples e.g. $(A,\cdot,C,\cdot)$, as might be encountered with common objective functions for Eq.~\ref{Eq:F_Polynomial},~\ref{Eq:P_Polynomial}. 
\item [(P4)] An efficient algorithm for computing achievable partial tuples optimal for some objective function.
\end{description}

Our main technical advances are precisely the resolution of problems (1-4). We describe in a simple and intuitive manner the set of achievable $(A,B,C,D)$, and provide efficient algorithms for solving what has traditionally been the hardest aspects of composite gate design. In particular, a beautiful connection is made with the historic field of discrete-time signal processing that allows allows us to inherit much of its prior work in polynomial design. In this manner, the inspired art of composite gates is transformed into a systematic science. Optimal composite gates are simply polynomials optimal for the objective function, and these polynomials can be found efficiently.

\subsection{Characterization of quantum response functions}
\label{Sec:characterization}
We characterize here achievable choices of quantum response functions $(A,B,C,D)$ in a manner independent of $\vec\phi$, hence resolving problem (\textbf{P1}). By providing insight into the forms of possible $\hat{U}(\theta)$, we also obtain a quantitative explanation for the remarkable versatility of composite gates. Achievability constraints on the polynomials $(A,B,C,D)$ are as follows:
\begin{thm}[Achievable tuples]
\label{Thm:achievableABCD}
A tuple of polynomials $(A,B,C,D)$ of degree at most $L$ is achievable iff all the following are true:
\\
(1) $A,B,C,D$ are real.
\\
(2) $A(1)=1$\space or\space $B(1)=0$. 
\\
(3) $
\Bigl\{
\begin{array}{@{}ll}
 A,B,C,D \text{ are odd} , &   L\text{ odd}, \\
 A,B \text{ are even and } C,D \text{ are  odd},&  L\text{ even}.
\end{array}
$
\\
(4) $1= 
\Bigl\{
\begin{array}{@{}ll}
 A(x)^2+ B(x)^2+C(y)^2+ D(y)^2 , & L\text{ odd}, \\
 A(x)^2+ B(x)^2+x^2C(y)^2+ x^2D(y)^2,& L\text{ even}.
\end{array}
$
\end{thm}
\begin{proof}
In forward direction, (1) and (3) are true by applying the trigonometric substitution $x^2+y^2=1$ in Eq.~\ref{U_Expansion} and collecting coefficients of $\hat{\mathbb 1},\hat\sigma_{x,y,z}$. (2) is true as $\hat{U}(0)=\hat{\mathbb 1}$ in Eq.~\ref{U_Expansion}. (4) is true as $\hat{U}$ is unitary so $\hat {U}^{\dagger}\hat{U}=\hat{\mathbb 1}$ and $\frac{1}{2}\text{Tr}[\hat {U}^{\dagger}\hat{U}]$ evaluated via Eq.~\ref{Eq:U_Polynomial} produces
\begin{align}
\label{Eq:Unitarity}
1 
=
\Bigl\{
\begin{array}{@{}ll}
A(x)^2+ B(x)^2+C(y)^2+ D(y)^2 , & L\text{ odd}, \\
A(x)^2+ B(x)^2+x^2C(y)^2 + x^2D(y)^2,& L\text{ even}.
\end{array}
\end{align}
In the reverse direction, we need to show that any $(A,B,C,D)$ satisfying (1-4) is achievable in the sense of Def.~\ref{Def:achievable}. We leave these steps to Lem.~\ref{Thm:phifromABCD}.
\end{proof}

Conditions (1-4) for achievable $(A,B,C,D)$ appear fairly general, which allows for great flexibility in choosing arbitrary response functions. They are also understandable and intuitive. A characterization of achievable partial tuples is also useful. Not all quadratures of $\hat{U}(\theta)$ might be relevant to an objective function, and optimizing over a subset $(A,B,C,D)$ could be easier. In the following, we examine how the unitarity constraint of condition (4) is weakened for all possible $2$-partial tuples.
\begin{thm}[Achievable $2$-partial tuples]
\label{Thm:achievablePartialABCD}
Assuming $A,B,C,D$ satisfy conditions (1-3) of Thm.~\ref{Thm:achievableABCD},\\
(1) $(A,\cdot,C,\cdot)$, $(A,\cdot,\cdot,C)$ is achievable iff \\
\phantom{\quad}(1a)
$
\forall \theta\in\mathbb{R},
\Bigl\{
\begin{array}{@{}ll}
 A(x)^2+C(y)^2\le 1  , &   L\text{ odd}, \\
 A(x)^2+x^2C(y)^2\le 1  ,&  L\text{ even}.
\end{array}
$\\
(2) $(\cdot,B,C,\cdot)$, $(\cdot,B,\cdot,C)$ is achievable iff \\
\phantom{\quad}(2a)
$
\forall \theta\in\mathbb{R},
\Bigl\{
\begin{array}{@{}ll}
 B(x)^2+C(y)^2\le 1  , &   L\text{ odd}, \\
 B(x)^2+x^2C(y)^2\le 1  ,&  L\text{ even}.
\end{array}
$\\
(3)$(A,B,\cdot,\cdot)$ is achievable iff \\
\phantom{\quad}(3a) $\forall \theta\in\mathbb{R},\;A(x)^2+B(x)^2\le 1$, and\\
\phantom{\quad}(3b) $\forall x\ge 1,\;A(x)^2+B(x)^2\ge 1$, and\\
\phantom{\quad}(3c) $\forall L\text{ even},\forall x\ge\tjyoder{0}, A(i x)^2+B(i x)^2\ge 1$.
\\
(4)$(\cdot,\cdot,C,D)$ is achievable iff \\
\phantom{\quad}(4a) $
\forall \theta\in\mathbb{R}, \Bigl\{
\begin{array}{@{}ll}
C(y)^2+D(y)^2\le 1  , &   L\text{ odd}, \\
x^2C(y)^2+x^2D(y)^2\le 1  ,&  L\text{ even},
\end{array}
$ and\\
\phantom{\quad}(4b) $\forall L\text{ odd}, y\ge 1,\;C(y)^2+D(y)^2\ge 1$.
\end{thm}
\begin{proof}
In the forward direction, all the (a) conditions are true from Eq.~\ref{Eq:Unitarity} using the fact that $A,B,C,D$ are all real, hence their squares are positive. (3b) is true by considering Eq.~\ref{Eq:Unitarity} with the substitution $x=\sqrt{\lambda}, \;y = \sqrt{1-\lambda}$, and computing $1-A^2(\sqrt{\lambda})-B^2(\sqrt{\lambda})=\cdots$. \ghl{Note that the $x,y$ here are complex}. Using the odd/even symmetry of $C,D$, the RHS factorizes into a positive term times $(1-\lambda)$ or $\lambda(1-\lambda)$. This is negative $\forall \lambda \ge 1$ so $A^2(\sqrt{\lambda})+B^2(\sqrt{\lambda})\ge 1$. (3c) is similarly proven by considering $\lambda\le 0$. The RHS factorizes into $\lambda(1-\lambda)$ and a positive term. (4b) is proven with the substitution $x=\sqrt{1-\lambda}, \;y = \sqrt{\lambda}$ and by considering $\lambda\ge 1$. In the reverse direction, we need to show that assuming these conditions enable the computation of an achievable $(A,B,C,D)$. We leave these steps to Lems.~\ref{Thm:CDfromAB}.~\ref{Thm:BDfromAC}.
\end{proof}
\ghl{Note that $C,D$ are interchangeable in Thm.~\ref{Thm:achievablePartialABCD} as their constraints in Thm.~\ref{Thm:achievableABCD} are identical}. We also characterize all possible $3$-partial tuples.
\begin{thm}[Achievable $3$-partial tuples]
\label{Thm:achievable3PartialABCD}
Assuming $A,B,C,D$ satisfy conditions (1-3) of Thm.~\ref{Thm:achievableABCD}, the following are achievable under their respective conditions\\
(1) $(A,\cdot,\cdot,\cdot)$ iff $\forall \theta\in\mathbb{R}, A(x)^2\le 1$.
\\
(2) $(\cdot,B,\cdot,\cdot)$ iff $\forall \theta\in\mathbb{R}, B(x)^2\le 1$.
\\
(3) $(\cdot,\cdot,C,\cdot)$ iff $\forall \theta\in\mathbb{R}, \Bigl\{
\begin{array}{@{}ll}
C^2(y)\le 1  , &   L\text{ odd}, \\
 x^2C^2(y)\le 1  ,&  L\text{ even}.
\end{array}$
\\
(4) $(\cdot,\cdot,\cdot,D)$ iff $\forall \theta\in\mathbb{R}, \Bigl\{
\begin{array}{@{}ll}
D^2(y)\le 1  , &   L\text{ odd}, \\
 x^2D^2(y)\le 1  ,&  L\text{ even}.
\end{array}$
\end{thm}
\begin{proof}
The forward direction follows by definition and from Eq.~\ref{Eq:Unitarity} where $A,B,C,D$ are all real $\forall \theta\in\mathbb{R}$, hence their squares are positive. The reverse direction is true from setting the unspecified polynomial to $0$ in one of the $2$-partial tuples (1), \tjyoder{(2)}  in Thm.~\ref{Thm:achievablePartialABCD}.
\end{proof}

These simple characterizations show how one can in principle encode almost any arbitrary desired function into quadratures of $\hat{U}(\theta)$. Consider $(A,\cdot,\cdot,\cdot)$, which aside from symmetry and $A(1)=1$, only needs to satisfy $\forall |x|\le 1,\;A^2(x)\le 1$. The famous Stone-Weierstrass theorem~\cite{Stone1948} assures us that $A(x)$ of sufficiently large degree $L$ can approximate arbitrarily well any \emph{arbitrary} continuous real function that satisfies these constrains on the interval $|x|\le 1$. This ability to create almost arbitrary quantum response functions helps explain the applicability of composite gates to many diverse problems.

\subsection{Implementation of quantum response functions}
\label{Sec:Implementation}
Unleashing the potential of arbitrarily sophisticated choices of achievable $(A,B,C,D)$ requires an efficient computation of their implementation $\vec \phi$. It is clear that the a random search is wholly inadequate as the degree of $L$ could be very large. Nevertheless, achievability leads to a certain structure that resolves this problem (\textbf{P2}). This 
is encapsulated in the following lemma, which is proven constructively and furnishes the reverse direction proof of Thm.~\ref{Thm:achievableABCD}.
\begin{lemma}[Optimal quantum response compilation]
\label{Thm:phifromABCD}
Exactly $L$ phases $\vec \phi\in \mathbb{R}^L$ are required to implement an achievable $(A,B,C,D)$ of degree at most $L$, and these $L$ phases can be computed in time $\mathcal{O}(\text{poly}(L))$.
\end{lemma}
\begin{proof}
A minimum of $L$ phases $\vec{\phi}$ are required to implement a given $(A,B,C,D)$ of degree at most $L$ as each application of $\hat R_{\phi_k}(\theta)$ only increases the degree of $(A,B,C,D)$ by one. We now show that $(A,B,C,D)$ can be implemented with at most $L$ phases $\vec{\phi}$. Due to the even/odd symmetry of real $A,B,C,D$ from Thm.~\ref{Thm:achievableABCD} conditions (1) and (3), we can compute its unique phase sum representation in Eq.~\ref{U_Expansion} via the invertible transformation
\begin{align}
\label{Eq:PhiFromPoly}
\Phi_{L,j} =  i^{j}\sum_{n=0}^L
\begin{cases} 
(i c_n-d_n) \binom {\lfloor(L-n)/2\rfloor} {(j-n)/2}, &j\, \text{odd}, \\ 
(a_n+ib_n) \binom {(L-n)/2} {j/2}, &j\,\text{even}.
\end{cases}
\end{align}
Let us take the ansatz $\hat{U}(\theta)=\hat{R}_{\phi_L}(\theta)\hat{V}(\theta)$ where $\hat{V}(\theta)$ is unitary and $\hat{V}(0)=\mathbb 1$ as $(A,B,C,D)$ represents a unitary from Thm.~\ref{Thm:achievableABCD} condition (4). Thus $\hat{V}(\theta)$ also has a phase sum representation $\Phi_{L-1,j}$. These two phase sums are related by the linear map of Eq.~\ref{Eq:PhaseSumRecurrance}, with inverse
\begin{align}
\label{Eq:PhiInverse}
\Phi_{L-1,j}=& \sum_{k=0}^j\Phi_{L,k}
\begin{cases}
\textstyle -e^{-(-1)^{j}i\phi_{L}},& j+k\text{ odd}, \\
\textstyle 1,& j+k\text{ even}.
\end{cases}
\end{align}
By choosing
\begin{align}
\label{Eq:choosephi}
e^{i\phi_L} = \frac{\sum_{k=1\text{ odd}}^L\Phi_{L,k}}{\sum_{k=0\text{ even}}^L\Phi_{L,k}}=\frac{c_{2\lceil L/2\rceil-1}+id_{2\lceil L/2\rceil-1}}{(-1)^{\lceil L/2\rceil}(a_L+ib_L)},
\end{align}
we satisfy the necessary condition $\Phi_{L-1,L}=0$ from Eq.~\ref{Eq:PhaseSumRecurrance}. In particular, $\phi_L$ is real, as Eq.~\ref{Eq:Unitarity} has the trailing term $((a_L^2+b_L^2)-(c_{2\lceil L/2\rceil-1}^2+d_{2\lceil L/2\rceil-1}^2))\sin^{2L}{\left(\theta/2\right)}=0$. Hence the RHS of Eq.~\ref{Eq:choosephi} has absolute value $1$. By recursively reducing the degree of $\hat{V}(\theta)$, we obtain all $L$ phases $\vec{\phi}$. The terminal case at $L=1$ must be consistent with Eq.~\ref{Eq:PhaseSumRecurrance} where $\Phi_{0,0}=1$. When evaluated with Eq.~\ref{Eq:PhiFromPoly},~\ref{Eq:PhiInverse}, this is satisfied only if $A(1)=1$ (Thm.~\ref{Thm:achievableABCD} condition (2)), which is true for achievable $(A,B,C,D)$. 
All steps in this procedure can be computed in time $\mathcal{O}(\text{poly}{(L)})$, and there are only $L$ recursions, leading to a runtime of $\mathcal{O}(\text{poly}{(L)})$. 
\end{proof}

\subsection{Computation of quantum response functions}
\label{Sec:Computation}
A consequence of Lem.~\ref{Thm:phifromABCD} is that designing a composite gate is no more difficult than finding the $(A,B,C,D)$ to describe the quantum response function $\hat{U}(\theta)$. Optimizing $(A,B,C,D)$ for some objective function is far more intuitive than the prior art of a random search over $\vec\phi$. However, this still is a difficult problem. The unitary constraint Eq.~\ref{Eq:Unitarity} represents a system of quadratic multinomial equations that would have to be solved at each step of the optimization to obtain an achievable  $(A,B,C,D)$. Solving such systems is in general an NP-complete task. This is the essence of problem (\textbf{P3}): it would be much easier to optimize a subset of $(A,B,C,D)$, and doing so is often the problem of practical interest anyway. 

This subset optimization is illustrated by the response functions $F_{\chi}(\theta)$, $p(\theta)$ of
Eq.~\ref{Eq:F_Polynomial},\ref{Eq:P_Polynomial} which depend on only two polynomials. Optimizing just these for some objective function offers more freedom as the unitary constraint Eq.~\ref{Eq:Unitarity} is weakened to that of Thm.~\ref{Thm:achievablePartialABCD}. Ultimately, we must compute some achievable $(A,B,C,D)$ from a partial specification in order to find the phases $\vec\phi$. 

Fortunately, the structure of achievable partial tuples can be exploited to derive algorithms \ghl{analogous to prior art~\cite{Shinnar1989} based on polynomial sum-of-squares problems~\cite{Marshall2008}, but specialized to the symmetries of Thm.~\ref{Thm:achievablePartialABCD}}. We present results for $(A,B\cdot,\cdot)$, $(A,\cdot,C,\cdot)$ of odd degree and show how they apply to all achievable $2$-partial tuple. \ghl{As these primarily serve to show that the necessary conditions in Thms.~\ref{Thm:achievableABCD},~\ref{Thm:achievablePartialABCD},~\ref{Thm:achievable3PartialABCD} are also sufficient, the details of the proofs for Lems.~\ref{Thm:CDfromAB},~\ref{Thm:BDfromAC}, which also furnish constructive algorithms for computing $(A,B,C,D)$ from partial tuples, may be skipped by the casual reader.}
\begin{lemma}[Transition probability sum-of-squares]
\label{Thm:CDfromAB}
$\forall$ $2$-partial tuples $(A,B,\cdot,\cdot)$ of odd degree at most $L$ that satisfy conditions \tjyoder{(1-3) of Thm.~\ref{Thm:achievableABCD}} and (3a, 3b) of Thm.~\ref{Thm:achievablePartialABCD}, 
$\exists$ achievable $(A,B,C,D)$ of degree at most $L$ that can be computed in time $\text{poly}(L)$.
\end{lemma}
\begin{proof}
Consider the polynomial of degree at most $L$
\begin{align}
f(\lambda)=1-A^2(\sqrt{1-\lambda})-B^2(\sqrt{1-\lambda}),\quad \ghl{\lambda\in \mathbb{R}},
\end{align}
with roots \ghl{$S=\{s\;|\;f(s)=0\}\in\mathbb{C}^{L}$ \tjyoder{($S$ contains duplicates if a root is degenerate). Since $A$, $B$ are odd polynomials, $f(\lambda)$ is real for all real $\lambda$.} Because $f(\lambda)$ is real, complex roots $s,s^{*}$ occur in pairs. Thus we can group subsets of $S$ without loss of \tjyoder{generality} as:}
\begin{align}
S_0&=\{s\in S\;|\;s=0\},\quad S_{\text{c}}=\{s\in S\;|\;\Imp[s]>0\}, \\ \nonumber
S_{\text{r}}&=\{s\in S\;|\;\Rep[s]\neq0\land\Imp[s]=0\}.
\end{align}
\ghl{Observe that $S_{0,r}$ are real, and $S_c$ is complex}. Thus
\begin{align}
f(\lambda)=\resizebox{.73\columnwidth}{!}{%
$\displaystyle K^2 \lambda^{|S_0|}\prod_{s\in S_{\text{r}}}(\lambda-s)\prod_{s\in S_{\text{c}}}\left((\lambda-\Rep[s])^2+\Imp[s]^2\right)$
},
\end{align}
with scale constant $K\in\mathbb{R}$.
Using (3b), $f(\lambda)\le 0,\; \forall \lambda\le 0$. Hence, all \tjyoder{negative} roots in $S_{\text{r}}$ occur with even multiplicity. Using (3a), $f(\lambda)\in [0,1],\; \forall \lambda\in[0,1]$. As $f(\lambda)$ changes sign at $\lambda=0$, $|S_0|$ is odd. Using the oddness of $A,B$, $f(\lambda)\ge 1,\; \forall \lambda\ge 1$. Since $ f(\lambda)\ge0,\;\forall\lambda\ge0$, all \tjyoder{positive} roots in $S_{\text{r}}$ occur with even multiplicity. Thus, all real roots excluding $s=0$ occur with even multiplicity. By repeated application of the two-squares identity
\begin{align}
\label{Eq:TSI}
(r^2+s^2)(t^2+u^2)=(r t \pm s u)^2+(r u \mp s t)^2,
\end{align}
the complex factors \tjyoder{can be simplified like}
\begin{align}
\prod_{s\in S_{\text{c}}}\left((\lambda-\Rep[s])^2+\Imp[s]^2\right)=g^2(\lambda)+h^2(\lambda),
\end{align}
where $g,h$ are real polynomials in $\lambda$. Thus $f(\lambda)=C^2(\sqrt{\lambda})+D^2(\sqrt{\lambda})$ where
\begin{align}
\label{Eq:CDfromAB}
\begin{Bmatrix}
C(y) \\ D(y)
\end{Bmatrix}
=
\left(K y^{|S_0|}\prod_{s\in S_{\text{r}}}(y^2-s)^{\frac{1}{2}}\right)
\begin{Bmatrix}
g(y^2) \\ h(y^2)
\end{Bmatrix},
\end{align}
and $C,D$ are odd real polynomials of degree at most $L$. Note that different choices of signs Eq.~\ref{Eq:TSI} generates a finite number of different valid solutions. Computing the roots of $f(\lambda)$ is the most difficult step of this algorithm, but can be done in time $\mathcal{O}(\text{poly}(L))$~\cite{Neff1996}.
\end{proof}
The proof for even $L$, and tuples $(\cdot,\cdot,C,D)$ carries through with minor modification. The stated conditions in Thm.~\ref{Thm:achievablePartialABCD} guarantee that the various factors of $\lambda$, $(1-\lambda)$ necessary for the correct symmetry of the unspecified polynomials occur with the right multiplicity, and that all other real roots occur with even multiplicity. Some additional processing for the $(\cdot,\cdot,C,D)$ case is required as the output $(A,B,C,D)$ is not guaranteed to satisfy $A(1)=1$. However, $A(1)^2+B(1)^2=1$ is still true so by computing $\gamma=\text{Arg}[A(1)+iB(1)]$, we can form an achievable $(A \cos{\gamma}+B \sin{\gamma},B \cos{\gamma}-A \sin{\gamma},C,D)$.

We now present the analogous algorithm for $(A,\cdot,C,\cdot)$.
\begin{lemma}[Fidelity response sum-of-squares]
\label{Thm:BDfromAC}
$\forall$ $2$-partial tuples $(A,\cdot,C,\cdot)$ of odd degree at most $L$ that satisfy condition\tjyoder{s (1-3) of Thm.~\ref{Thm:achievableABCD} and} (1a) of Thm.~\ref{Thm:achievablePartialABCD}, $\exists$ achievable $(A,B,C,D)$ of degree at most $L$ that can be computed in time $\text{poly}(L)$.
\end{lemma}
\begin{proof}
With the Weierstrass substitution \ghl{$\forall  t\in\mathbb{R},\;x=(1-t^2)/(1+t^2), y=2t/(1+t^2)$}, define the real polynomials
\begin{align}
(\tilde A(t),\tilde B(t),\tilde C(t),\tilde D(t))= (1+t^2)^L (A,B,C,D).
\end{align}
\ghl{These polynomials have extremely useful symmetries which we indicate with angled brackets $\langle\cdot\rangle$. $\langle\tilde A\rangle=\langle \tilde B\rangle=\langle E\tjyoder{N}\rangle$ are Even $\langle\text{E}\rangle$ aNtipalindromes $\langle\text{N}\rangle$ while $\langle\tilde C\rangle=\langle \tilde D\rangle=\langle OP\rangle$ are Odd $\langle\text{O}\rangle$ Palindromes $\langle\text{P}\rangle$}. Antipalindromes satisfy $\tilde A(t)=-t^{2L} \tilde A(t^{-1})$ whereas palindromes satisfy $\tilde C(t)=t^{2L} \tilde C(t^{-1})$. Note that $\langle\text{E}\rangle$,$\langle\text{O}\rangle$ and $\langle\text{P}\rangle$,$\langle\text{N}\rangle$ polynomials with multiplication form a group isomorphic to $Z_2 \times Z_2$. For example, $\langle\text{EN}\rangle\langle\text{OP}\rangle=\langle\text{ON}\rangle$. 

Consider the positive, \tjyoder{palindromic} polynomial
\begin{align}
\label{Eq:BDfromACpositivepolynomial}
\tilde f(t)&
\resizebox{.8\columnwidth}{!}{$\displaystyle=(1+t^2)^{2L}-\tilde A^2(t)-\tilde  C^2(t)=K^2\prod^{}_{s\in S}(t-s)$},
\end{align}
with scale constant $K\in\mathbb{R}$, and roots \ghl{$S=\{s\;|\:\tilde f(s)=0\}\in\mathbb{C}^{4L-|S_0|}$}, where $|S_0|$ is the multiplicity of the zero roots. \tjyoder{Note the degree of $\tilde f(t)$ is $4L-|S_0|$, not $4L$, because the first $|S_0|$ coefficients being zero implies the last $|S_0|$ are as well.} Due to the $\langle EP\rangle$ symmetry of $\tilde f(t)$, $\forall$ roots $s\neq0$, $\exists$ roots $s^{*}$,$-s$, and $s^{-1}$. Thus we group subsets of these roots without any loss of information as follows:
\begin{align}
S_0 &= \{s\in S\;|\; s = 0\}, \quad S_1 = \{s\in S\;|\; s = 1\},\\ \nonumber
S_r &= \{s\in S\;|\; \Rep[s] > 1 \land \Imp[s]=0\},\\ \nonumber 
S_i &= \{s\in S\;|\; \Rep[s] = 0 \land \Imp[s]=1\},\\ \nonumber 
S_\iota &= \{s\in S\;|\;\Rep[s] = 0 \land \Imp[s]>1\},\\ \nonumber 
S_u &= \{s\in S\;|\; |s| = 1 \land 0<\text{Arg}[s]<\pi/2\},\\ \nonumber 
S_c &= \{s\in S\;|\; |s| > 1 \land 0<\text{Arg}[s]<\pi/2\}. 
\end{align}
\ghl{
Observe that $S_{0,1,r}$ are real, $S_{i,\iota}$ are imaginary and $S_{u,c}$ are complex}.
From the real roots, we construct the factor
\begin{align} \nonumber
f_{\text{r}} &=\textstyle t^{\frac{|S_0|}{2}}  (t^2-1)^{\frac{|S_1|}{2}} \prod_{s\in S_r} \left(t^4-t^2(s^2+s^{-2})+1\right)^{\frac{1}{2}}, \\
\langle f_{\text{r}}\rangle &= \langle\text{OP}\rangle^{\frac{|S_0|}{2}}\langle\text{EN}\rangle^{\frac{|S_1|}{2}}\langle\text{EP}\rangle^{\frac{|S_r|}{2}},
\end{align}
The positiveness of $\tilde f(t)$ means that all real factors have even multiplicity. Thus $f_{\text{r}}$ is a polynomial. From the complex roots, we form
\begin{align}
f_{\text{i}} =& \left((t^2-1)^2+(2t)^2\right)^{\frac{|S_i|}{2}}, \\ \nonumber
f_{\iota} =&\textstyle\prod_{s\in S_\iota} (t^2-1)^2+(t(\Imp[s]+\Imp[s]^{-1}))^2, \\  \nonumber
f_{\text{u}} =&\textstyle\prod_{s\in S_\text{u}}(t^2-1)^2+(2t \sin{(\text{Arg}[s])})^2,\\ \nonumber
f_{\text{c}}=&\textstyle\prod_{s\in S_\text{c}}\left( t^4-t^2(|s|^{-2}+4\sin^2{(\text{Arg}[s])}+|s|^{2})+1\right)^2\\\nonumber
&\textstyle+ \left(2(t^3-t)\Imp[s]\left(1+|s|^{-2}\right)\right)^2.
\end{align}
 The symmetry of terms under the squares is one of $\langle\text{EP}\rangle,\langle\text{EN}\rangle,\langle\text{OP}\rangle,\langle\text{ON}\rangle$, and occur in a combination that forms a group under repeated application of the two-squares identity of Eq.~\ref{Eq:TSI}. Thus we can construct
\begin{align}
f_i f_{\iota}f_{\text{u}}f_{\text{c}} & =g^2+h^2,  \\ \nonumber
\langle g \rangle &= \langle\text{EN}\rangle^{\frac{|S_i|}{2}+|S_\text{u}|+|S_\iota|},\quad \langle h\rangle = \langle\text{OP}\rangle^{\frac{|S_i|}{2}+|S_\text{u}|+|S_\iota|}, \\ \nonumber
\tilde f(t)&=(K f_{\text{r}}g)^2+(K f_{\text{r}}h)^2.
\end{align}
For some combinations of multiplicities, this decomposition will not produce polynomials with the symmetry $\langle\text{EN}\rangle, \langle\text{OP}\rangle$ required by $\tilde B$, $\tilde D$. However, summing the multiplicities of these roots shows that $|S_i|$ is even and that such combinations do not exist.
From this decomposition, we compute $B(x), D(y)$ using
\begin{align}
\label{Eq:computeBD}
b_{k} =& \textstyle \frac{1}{2^{L}}\sum_{n=0}^{L}\tilde{b}_{2n}\left[\sum_{m=0}^n (-1)^m\binom{n}{m}\binom{L-n}{k-m}\right], \\ \nonumber
d_{2k+1} =& \textstyle  \frac{(-1)^{k}}{2^{L}}\sum_{n=0}^{L-1}\tilde{d}_{2n+1}\left[\sum_{m=0}^{n}\sum_{p=0}^{\lfloor L/2\rfloor}(-1)^{m}\binom{p}{k}\right. \\ \nonumber
&\textstyle\left.\binom{L-n-1}{2p-m}\binom{n}{m}\right].
\end{align}
As with Lem.~\ref{Thm:CDfromAB}, different choices of signs in the two-squares identity lead to multiple valid solutions. Computing the roots of $\tilde f(t)$ is the still most difficult step, but can be done in time $\mathcal{O}(\text{poly}(L))$.
\end{proof}
The case of even $L$ replaces Eq.~\ref{Eq:BDfromACpositivepolynomial} with $\tilde f(t)=(1+t^2)^{2L}-\tilde A^2(t)-((1-t^2)/(1+t^2))^2\tilde C^2(t)$ and we find $\tilde B$ with $\langle \text{EP}\rangle$ and $(1-t^2)\tilde D$ with $\langle \text{ON}\rangle$ symmetry . A similar root counting argument guarantees the existence of such solutions. The coefficients of $B(x), D(y)$ are then computed also using Eq.~\ref{Eq:computeBD}. This procedure carries through without modification for the other tuples (1), (2) of Thm.~\ref{Thm:achievablePartialABCD}.

\subsection{Selection of quantum response functions}
\label{Sec:Selection}
It should be clear that optimal composite gate design is a systematic process no more difficult than choosing one or two polynomials optimal for some objective function. Nevertheless, problem (\textbf{P4}) is that computing these optimal polynomials could still be a difficult task. However, the constraints on achievable partial tuples in Thms.~\ref{Thm:achievablePartialABCD}, \ref{Thm:achievable3PartialABCD} seem fairly lax, which lends hope that this could be done efficiently. In fact these constraints are consistent with textbook problems in approximation theory~\cite{Meinardus2012}.

It is at this point where a close connection with discrete-time signal processing~\cite{Oppenheim2010} is made. Efficient algorithms~\cite{Mcclellan1973,karam1999,Grenez1983,Lang2000,hofstetter1971} for designing polynomials optimal for arbitrary objective functions under a variety of optimality criteria have been extensively studied for finite-impulse response filters~\cite{Harris1978}. We thus inherit much of this machinery, and in many cases, existing polynomials consistent with achievability have already been found and are directly transferable. 

A most common optimality criterion is the Chebyshev norm: Let $P_\text{o}(x)$ be the objective function, with continuous weight function $W(x)>0$, to be approximated by a polynomial $P(x)$ of degree $L$ on a bounded subset $\mathcal{B}$ of the closed interval $\mathcal{B}\subset[-1,1)$ with the smallest Chebyshev error norm 
\begin{align}
\label{Eq:ChebyshevError}
\epsilon = \max_{x\in \mathcal{B}}{\left|W(x)\left(P(x)-P_{\text{o}}(x)\right)\right|}.
\end{align}
The unique best approximation can be computed efficiently by Remez-type exchange algorithms~\cite{Fraser1965}. Many variants exist such as where $P(x)$ is a trigonometric polynomial~\cite{Mcclellan1973}, bounded~\cite{Grenez1983}, subject other unary or linear constraint~\cite{Lang2000}, and even complex~\cite{karam1999}. Linear programming methods~\cite{Lang2000} provide an alternate solution. Efficient algorithms for other optimality criteria such as least squares are also available~\cite{lim1992weighted,vaidyanathan1987eigenfilters}.

These algorithms efficiently solve the problem of optimization over achievable quantum response functions $\hat{U}(\theta)$ where the objective functions are $2$-partial or $3$-partial tuples. Optimization for a $3$-partial objective function involves a single quadrature from $(A,B,C,D)$ together with a single real objective function $P_{\text{o}}(\theta)$. Thus we optimize over $P(\theta)$ for $P_{\text{o}}(\theta)$ in Eq.~\ref{Eq:ChebyshevError} subject to the constraints of Thm.~\ref{Thm:achievable3PartialABCD} for the corresponding quadrature. The slightly more complicated $2$-partial case instead specifies two quadratures and real objective functions $P_{\text{o},1}(\theta),P_{\text{o},2}(\theta)$. Thus we define $P_{\text{o}}(\theta)=P_{\text{o},1}(\theta)+iP_{\text{o},2}(\theta)$, and optimize over $P(\theta)=P_{1}(\theta)+iP_{2}(\theta)$ for $P_{\text{o}}(\theta)$ subject to the constraints of Thm.~\ref{Thm:achievablePartialABCD} for the corresponding quadratures. Note that the unitarity inequality constraint poses no difficulty as $|P(\theta)|^2=P^2_{1}(\theta)+P^2_{2}(\theta)$. 

\subsection{The methodology of composite quantum gates}
\label{Sec:methodology}
Our efforts lead us to a methodology for the design of single spin quantum response functions $\hat{U}(\theta)$ through composite quantum gates built from a sequence of $L$ primitive gates all rotating by $\theta$, but each with its own phase $\vec{\phi}=(\phi_1,...,\phi_L)$. The procedure is systematic, flexible, and most importantly, provably efficient:
\begin{description}[align= left,noitemsep]
\item[Problem statement] Given $L\ge 1$ and objective function $\hat{U}_o(\theta)$ for either $3$-partial or $2$-partial tuples, find the composite quantum gate that implements through $\vec{\phi}$ the optimal $\epsilon$-approximation to $\hat{U}_0(\theta)$.
\item[Solution procedure]
\item [(S1)] Check that $\hat{U}(\theta)$ is consistent with achievability. \\ 
\phantom{\quad}--\textit{Satisfies conditions of Thms.~\ref{Thm:achievablePartialABCD},\ref{Thm:achievable3PartialABCD}.}
\item [(S2)] Choose optimality criterion. \\
\phantom{\quad}--\textit{The Chebyshev norm is most common.}
\item [(S3)] Execute polynomial optimization algorithm over achievable partial tuples. \\
\phantom{\quad}--\textit{Remez-type algorithms are efficient.}
\item [(S4)] Compute achievable tuple from partial tuple.\\ 
\phantom{\quad}--\textit{This can be done efficiently by Lems.~\ref{Thm:CDfromAB},~\ref{Thm:BDfromAC}.}
\item [(S5)] Compute phases $\vec{\phi}$. \\
\phantom{\quad}--\textit{This can be done efficiently by Lem.~\ref{Thm:phifromABCD}.}
\end{description}

\section{Examples}
\label{Sec:Example}
\ghl{
Using the methodology in Sec.~\ref{Sec:methodology}, composite quantum gates with response function $\hat{U}(\theta)$ that minimize the error with respect to arbitrary objective functions $\hat{U}_o(\theta)$ can be efficiently designed. We illustrate this process with three examples of independent scientific interest: compensated population inversion gates, compensated broadband NOT gates, and compensated narrowband quantum gates.

Population inversion gates rotate states $\ket{0}$ to $\ket{1}$ and vice-versa, and come in two flavors. The \emph{broadband} variant implements this rotation with high probability across the widest bandwidth of $\theta\in\mathcal{B}$, meaning that the transition probability response function $p(\theta)$ from Eq.~\ref{Eq:P_Polynomial} is close to $1$. The \emph{narrowband} variant instead implements this rotation with low probability so $p(\theta)\approx 0$, except at a single point $p(\pi)=1$. \tjyoder{We discuss optimal design of these gates in Section~\ref{Sec:PopInvGates}.} As closed-form solutions for these gates are already known, and used extensively in NMR spectroscopy, they help build familiarity with the methodology in Sec.~\ref{Sec:methodology} when it is used to solve open questions in the next two examples.

\emph{Broadband} compensated NOT gates implement the rotation $\hat{R}_0(\pi)$ with high fidelity over the widest bandwidth of $\theta$ parameters. Whereas population inversion gates only succeed on initial states $\ket{0}$ to $\ket{1}$, NOT gates apply a $\pi$ rotation with a known phase for all input states. Such gates have been extensively studied for applying uniform rotations in the presence of drive field inhomogeneities, particularly in quantum computing applications, and our methodology, \tjyoder{presented in Sec.~\ref{Sec:OBGates}}, solves open questions regarding the scaling of bandwidth with sequence length as well as their efficient synthesis.

A complementary design problem addressed in Sec.~\ref{Sec:SpatialGates} is that of \emph{narrowband} compensated quantum gates. These instead apply a desired arbitrary rotation $\hat{R}_0(\chi)$ at a single $\theta$ value, and the identity rotation elsewhere over widest bandwidth of $\theta$ parameters. Such gates are highly relevant to minimizing crosstalk in the selective addressing of spins in arrays, particularly when spin-spin distances are below the diffraction limit, as might be found in architectures for scalable architectures of ion-trap quantum computation.

}
\subsection{Composite population inversion gates}\label{Sec:PopInvGates}
\ghl{
Population inversion gates maximize the bandwidth $\mathcal{B}$ over which the transition probability response function $p(\theta)$ from Eq.~\ref{Eq:P_Polynomial} is close to $1$ for the broadband variant, or close to $0$ for the narrowband variant. Note that in both cases, perfect population inversion occurs at $\theta=\pi$ for $L$ odd, owing to the fact that $A(0)=0$. Moreover, the optimal polynomials and phases for both variants turn out to be related by a simple transformation, so it suffices for us to consider only the broadband case.

Composite gates with these properties have been studied extensively for nuclear magnetic resonance and quantum computing applications. One approach to obtaining broadband behavior is with the maximally flat ansatz $p(\theta)=1-\mathcal{O}((\theta-\pi)^{2n})$~\cite{Torosov2011}. This exponentially suppresses errors in the transition probability to order $n$, thus $p(\theta)\approx 1$ over a wide range of $\theta$. Remarkably, the $\vec{\phi}$ that implement this profile can be found in closed form~\cite{Vitanov2011} with optimal sequence lengths $L=n$. More recently, a second approach has emerged~\cite{Low2015}, motivated by the following observation: as the flat ansatz $p(\theta)=1-\mathcal{O}((\theta-\pi)^n)$ only increases bandwidth indirectly through the suppression order $n$, better results can be obtained by directly optimizing for bandwidth, while ensuring that the worst-case error $\mathcal{I}$ remained bounded.

The procedure of Sec.~\ref{Sec:methodology} for odd $L$ formalizes this task as a straightforward optimization problem:
\\
\textbf{(S1)} Choose the objective function $\forall \theta\in\mathcal{B}=\pi+[-|\mathcal{B}|/2,|\mathcal{B}|/2],\;\hat{U}_{\text{o}}(\theta)=0$ for the $(A,0,\cdot,\cdot)$ 2-partial tuple. Since $p(\theta)=1-A^2$ is close to $1$ over $\mathcal{B}$, the \tjyoder{unitarity} constraint $C^2+D^2=1-A^2$ implies that a rotation $R_\phi(\pi)$ is approximated over $\mathcal{B}$, with an unspecified phase $\phi=\text{Arg}[C+iD]$ that varies with $\theta$. As consistency with Thm.~\ref{Thm:achievableABCD} requires that $A(1)=1$, this implies that identity is applied at $\theta=0$, thus $\mathcal{B}$ must not contain $\theta=0$.
\\
\textbf{(S2)} Choose the Chebyshev optimality criterion, where the best $A$ solves the minimax optimization problem
\begin{align}
\label{Eq:PIBandwidthOptimization}
\epsilon&=\min_{A}\max_{\theta\in\mathcal{B}}|A(x)|,\quad 
\epsilon^2=\mathcal{I},
\end{align}
where the worst-case transition probability over $\mathcal{B}$ is $1-\mathcal{I}$.
\\
\textbf{(S3)} Find the function $A$ that solves Eq.~\ref{Eq:PIBandwidthOptimization}. For consistency with Thm.~\ref{Thm:achievableABCD}, the optimization is over real odd polynomials $A$ bounded by $\forall  |x|\lessgtr 1,\; |A(x)|\lessgtr 1$.
\\
\textbf{(S4)} Using Lem.~\ref{Thm:CDfromAB}, compute the achievable tuple $(A,0,C,D)$  from the partial specification $(A,0,\cdot,\cdot)$.
\\
\textbf{(S5)} Compute $\vec{\phi}$ from $(A,0,C,D)$ using Lem.~\ref{Thm:phifromABCD}.

The solution to \textbf{(S3)} is the Dolph-Chebyshev window function~\cite{dolph1946,Lynch1997} famous in discrete-time signal processing.
\begin{align}
\label{Eq:DCpoly}
\text{DC}_{L,\mathcal{I}}(y) &= \sqrt{\mathcal{I}} T_L\left(\beta_{L,\mathcal{I}}x\right), \; \beta_{L,\mathcal{I}}=T_{L^{-1}}(\mathcal{I}^{-1/2}),
 \end{align}
where $T_n(x)=\cos{\left(n \arccos {(x)}\right)}$ are Chebyshev polynomials. Note the ripples of $\text{DC}^2_{L,\mathcal{I}}(x)$ bounded by $\mathcal{I}$ in Fig.~\ref{DCPlot}. This is in contrast to monotonic increase of the limiting function, indicated by the subscript $\text{f}$,
\begin{align}
\label{Eq:LimitDC}
\text{DC}_{L,\text{f}}(x)=\lim_{\mathcal{I}\rightarrow 0}\text{DC}_{L,\mathcal{I}}(x)&=x^L,
\end{align}
which is maximally flat at $x=0$, but has significantly narrower bandwidth. Using $x=\cos{(\theta/2)}$, the bandwidth in $\theta$ coordinates is to order $\mathcal{O}(\mathcal{I}^{\frac{3}{2L}})$
\begin{align}
\label{Eq:DCbandwidth}
|\mathcal{B}_{}|= 2^{3-\frac{1}{L}}\mathcal{I}^{\frac{1}{2L}}, \quad
|\mathcal{B}_{\text{f}}|= 4\mathcal{I}_{\text{f}}^{\frac{1}{2L}},\quad\frac{\mathcal{I}}{\mathcal{I}_\text{f}}=4^{1-L}.
\end{align}
Given the same target bandwidth, the worst-case error of $\text{DC}_{L,\mathcal{I}}$ is exponentially smaller than $\text{DC}_{L,\text{f}}$. Note also the quadratic difference in the scaling with $L$ of the bandwidth over which $\text{DC}_{L,\mathcal{I}}$ does \emph{not} approximate $F(x)=0$.
\begin{align}
\label{Eq:DCbarbandwidth}
\bar{\mathcal B}=\tfrac{4\text{arcsech}{\sqrt{\mathcal{I}}}}{L}+\mathcal{O}{(\tfrac{1}{L^{3}})},
\;\; \bar{\mathcal B}_{\text{f}}=4\sqrt{\tfrac{\log{\frac{1}{\mathcal{I}}}}{L}}+\mathcal{O}{(\tfrac{1}{L^{\frac32}})}.
\end{align}
The ripples in the amplitude are a generic feature of best polynomial approximations to functions in the Chebyshev norm. By sacrificing flatness, much smaller absolute variations in error $\epsilon$ can be achieved over some specified bandwidth $\mathcal{B}$. This is a common theme that will be revisited in the subsequent example.
}
\begin{figure}
\includegraphics[width=1.0\columnwidth]{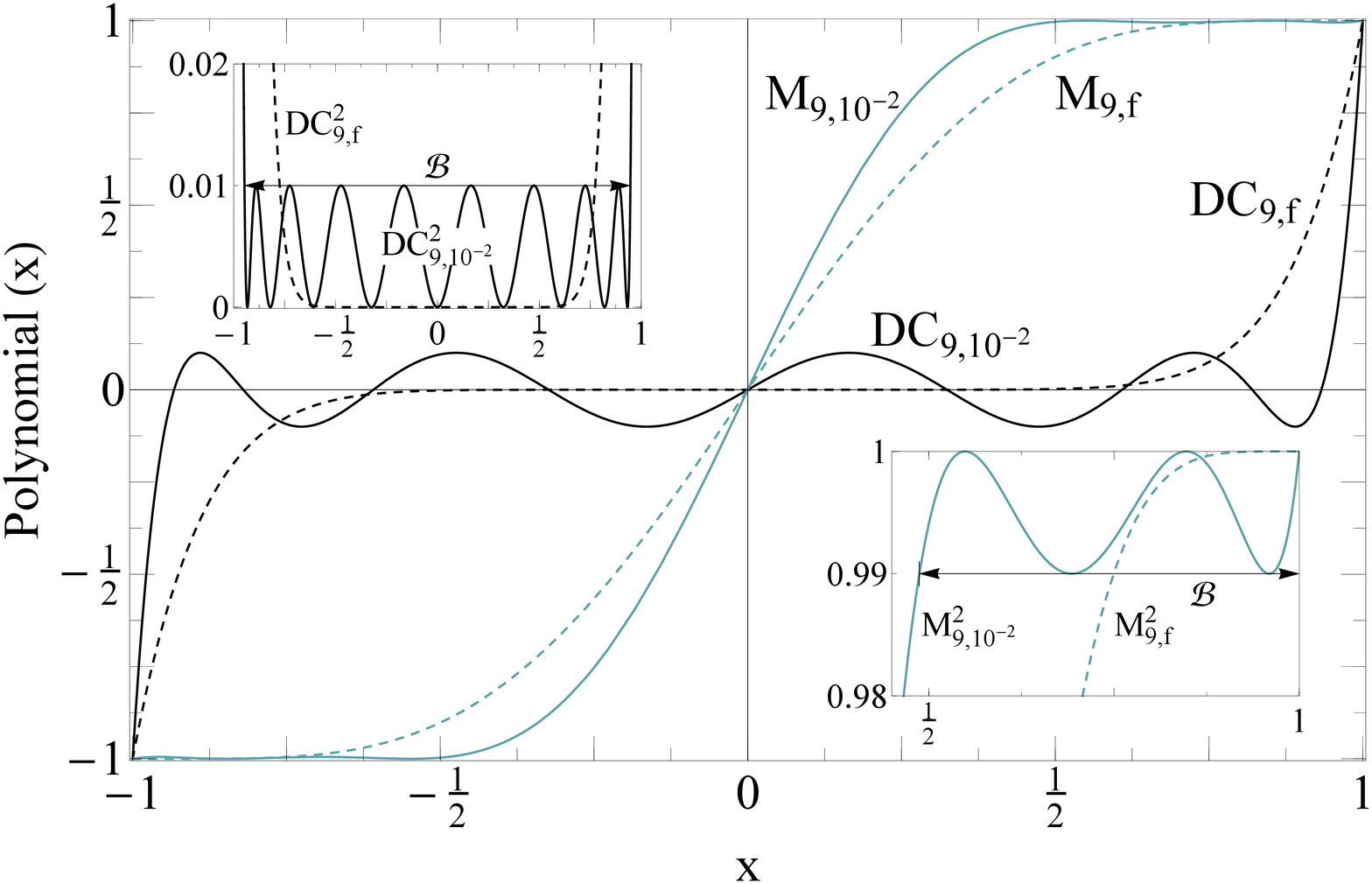}
\caption{\label{DCPlot}(color) 
$\text{DC}_{L,\mathcal{I}}$ (black), $\text{M}_{L,\mathcal{I}}$ (teal) polynomials plotted for $L=9$ and target worst-case infidelity $\mathcal{I}=10^{-2}$ (solid) and  $\mathcal{I}\rightarrow 0$ (dashed), indexed by $_\text{f}$. The observed ripples are a generic feature of bandwidth optimized polynomials, unlike those optimized for maximal flatness $\text{DC}_{\text{f}},\text{M}_{\text{f}}$. The inset plots their squares and defines the bandwidth $\mathcal{B}$ in $x$ coordinates.}
\end{figure}
\ghl{
Finding the phases that implement $(\text{DC}_{L,\mathcal{I}}(x),0,\cdot,\cdot)$ is then a straightforward computation through \textbf{(S4), (S5)}, and the results can be compared to the closed-form solutions from~\cite{Yoder2014,Low2015}: $\phi_{k}=\phi_{L-k+1}$ where $\phi_1 = 0$ and
\begin{align}
\phi_{k+1} &= \phi_{k} + 2 \tan^{-1}{\left[\tan{\left(\frac{2j k \pi}{L}\right)}\sqrt{1-\beta^{-2}_{L,\mathcal{I}}}\right]}.
\end{align}
The phases $\vec{\psi}$ for the narrowband variant $(\cdot,\cdot,\text{DC}_{L,\mathcal{I}}(x),0)$ are obtained by a simple `toggling' transformation~\cite{Wimperis1994} $\psi_j=-(-1)^j\phi_j-\sum_{k=1}^{j-1}2\phi_k$.
}
\subsection{Broadband compensated NOT gates}
\label{Sec:OBGates}
\ghl{
Broadband compensated NOT gates maximize the bandwidth $\mathcal{B}$ over which the fidelity response function with respect to the target gate $\hat{R}_0(\pi)$ is close to $1$. One option consistent with this goal is the choice of fidelity response functions $F_{\pi}(\theta)=1-\mathcal{O}((\theta-\pi)^{2n+2})$ that are maximally flat with respect to $(\theta-\pi)$. When the \emph{correction order} $n$ increases, deviations from $\theta=\pi$ are exponentially suppressed, resulting in improved approximations of the target gate over wider ranges of $\theta\in\mathcal{B}$. The central difficulty of this pursuit is finding the phases $\vec{\phi}$ that maximize $n$ for any given $L$. Unlike the population inversion gates of Sec.~\ref{Sec:PopInvGates}, this appears to be significantly more difficult; optimal length solutions for the $\vec{\phi}$ have only been found in closed-form for small $n\le 4$~\cite{Low2014}.

This problem has been attacked over the course of two decades, starting with Wimperis~\cite{Wimperis1994} who found the $\vec{\phi}$ in closed form for BB$_1$, a $L=5$ sequence with $n=2$. This was extended by Brown et. al.~\citep{Brown2004} with SK$n$ for arbitrary $L=\mathcal{O}(n^{3.09})$ through a recursive construction, and then by Jones~\cite{Jones2013,Husain2013} with F$_n$ to $L=\mathcal{O}(n^{1.59})$ in closed-form through sequence concatenation. The most recent effort~\cite{Low2014} proved a lower bound of $L=\Omega(n)$ and conjectured that the sequence BB$n$ ($W_n$ in~\cite{Husain2013}) with $L=2n+1$ is optimal through brute-force up to $L=25$. Using our methodology, we can easily prove this conjecture and efficiently compute its implementation $\vec\phi$. 

Moreover, our methodology enables a second option. Instead of optimizing for correction order, it is possible to directly minimize the worst-case infidelity $\mathcal{I}$, which is the experimental quantity of interest, over a target bandwidth $\mathcal{B}$. We find that doing so leads to an improvement in $\mathcal{I}$ that scales exponentially with $L$ over the maximally flat case. To prove these statements, we proceed with the design outline of Sec.~\ref{Sec:methodology} for odd $L$:
\\
\textbf{(S1)} Choose the objective function $\forall \theta\in\mathcal{B}=\pi+[-|\mathcal{B}|/2,|\mathcal{B}|/2],\;\hat{U}_{\text{o}}(\theta)=\hat{R}_0(\pi)=-i\sigma_x$ for the $(\cdot,\cdot,C,\cdot)$ 3-partial tuple. Provided that $\mathcal{B}$ does not contain the point $\theta=0$, this is consistent with the constraints of Thm.~\ref{Thm:achievable3PartialABCD}. This corresponds to finding a fidelity response function $F_{\pi}(\theta)=C^2(\sin{(\frac{\theta}{2})})$ that is close to $1$ across $\mathcal{B}$.
\\
\textbf{(S2)} The best fidelity response function for the maximally flat approach in prior art is obtained from the function $C$ that maximizes the correction order
\begin{align}
\label{Eq:NOTFlatOptimization}
n&=\max_{C}\{n\;|\;C(y)=1-\mathcal{O}((1-y)^{n+1})\},\\ \nonumber
\mathcal{I}&=1-\min_{\theta \in\mathcal{B}}F_\pi(\theta),\quad y\equiv\sin{(\theta/2)},
\end{align}
where $\mathcal{I}$ is the worst-case infidelity over the bandwidth $\mathcal{B}$. It is easy to verify that any such $C$ satisfies $F_\pi(\theta)=1-\mathcal{O}((\theta-\pi)^{2n+2})$. The more direct approach uses the Chebyshev optimality criterion, where the best $C$ solves the minimax optimization problem
\begin{align}
\label{Eq:NOTBandwidthOptimization}
\epsilon&=\min_{C}\max_{\theta\in\mathcal{B}}|C(y)-1|,\quad 
\mathcal{I}=1-(1-\epsilon)^2.
\end{align}
\\
\textbf{(S3)} Find the function $C$ that solves Eqs.~\ref{Eq:NOTFlatOptimization},~\ref{Eq:NOTBandwidthOptimization}. For consistency with Thm.~\ref{Thm:achievableABCD}, the optimization is over real, odd polynomials $C$ bounded by $|C(y)|\le 1\;,\forall y\in[-1,1]$.
\\
\textbf{(S4)} Using Lem.~\ref{Thm:BDfromAC}, compute the achievable tuple $(A,0,C,D)$  from the partial specification $(\cdot,0,C,\cdot)$.
\\
\textbf{(S5)} Compute $\vec{\phi}$ from $(A,0,C,D)$ using Lem.~\ref{Thm:phifromABCD}.

We now present the solutions to \textbf{(S3)} of this procedure. This is the most difficult step, as once $C$ is provided, the implementation $\vec{\phi}$ is a straightforward calculation. Eq.~\ref{Eq:NOTFlatOptimization} is solved by the the odd polynomial that satisfies the following $n+1$ independent linear constraints:
\begin{align}
\label{eq:Mflatconstraints}
C(1)=1,\quad \left.\frac{d^k}{d y^k}C(y)\right|_{y=1}=0,\quad k=1,2,...,n.
\end{align}
As a degree $L$ odd polynomial has $\frac{L+1}{2}$ free parameters, a degree $L=2n+1$ polynomial is necessary and sufficient. This is solved by the polynomial
\begin{align}
\label{Eq:NFlat}
\text{M}_{L,\text{f}}(y)=\sum_{j=0}^{(L-1)/2}\binom{L}{j}\left(\frac{1+y}{2}\right)^{L-j}\left(\frac{1-y}{2}\right)^j,
\end{align}
with an example $\text{M}_{9,\text{f}}$ plotted in Fig.~\ref{DCPlot}. The index $L$ indicates the degree, and the subscript $\text{f}$ indicates that this is a maximally flat polynomial. As $\text{M}_{L,\text{f}}(y)$ is monotonically decreasing from $y<1$,  the relation between infidelity $\mathcal{I}$ and bandwidth $\mathcal{B}$ is obtained by solving $\mathcal{I}=1-M^2_{L,\text{f}}(\cos{(|\mathcal{B}|/4)})$ to leading order:
\begin{align}
\label{eq:BBInfidelity}
\mathcal{I}=\left(\frac{|\mathcal{B}|}{8}\right)^{L+1}\frac{2^{L+5/2}}{\sqrt{\pi L}}\left[1+\mathcal{O}\left(\left(\frac{|\mathcal{B}|}{8}\right)^2+\frac{1}{L}\right)\right].
\end{align}
Thus given some target bandwidth $\mathcal{B}$ of high-fidelity operation, the composite quantum gate represented by $\text{BB}n=(\cdot,0,\text{M}_{2n+1,\text{f}}(y),\cdot)$ implements NOT with a worst-case fidelity that decreases exponentially with sequence length. This proves the $L=2n+1$ conjecture of~\cite{Low2014}.

The odd polynomials of degree $L$ that satisfy the Chebyshev error norm optimality criterion in Eq.~\ref{Eq:NOTBandwidthOptimization} can also be found. We label these polynomials $\text{M}_{L,\mathcal{I}}$, where $L$ indicates the degree, and $\mathcal{I}$ is the worst-case infidelity, which is directly related to the bandwidth $\mathcal{B}$. For $L=5$, we have a complicated looking expression
\begin{align}\nonumber
\label{Eq:N5polynomial}
\text{M}_{5,\mathcal{I}}&=\tfrac{(2 y_1+1) y^5-(4 y_1^3+3 y_1^2+2 y_1+1) y^3+(2 y_1^5+4 y_1^4+6 y_1^3+3 y_1^2) y}{2 y_1^3 (y_1+1){}^2},\\
\mathcal{I}&=
\tfrac{(y_ 1-1){}^3 (1+3 y_ 1+y_ 1^2){}^2 (1-2 y_ 1-4 y_ 1^2) (3+9 y_ 1+8 y_ 1^2){}^3}{3125 y_ 1^6 (1+y_ 1){}^4(1+2 y_ 1)^3},
\end{align}
parameterized implicitly through $y_1\in[\cos{(\pi/5)},1]$. For larger $L$, such as $\text{M}_{9,10^{-2}}$ in Fig.~\ref{DCPlot}, the $\text{M}_{L,\mathcal{I}}$ can always be computed numerically through the famous Parks--McClellan algorithm~\cite{Mcclellan1973} for finite impulse response filters. Remarkably, the Chebyshev error of this approximation problem is known~\cite{Eremenko2007}:
\begin{align}
\label{eq:OBInfidelity}
\epsilon&=
\left(\frac{1}{\sqrt{2\pi}}+o(1)\right)\frac{8\cos^2{(|\mathcal{B}|/8)}\tan^{L+1}{(|\mathcal{B}|/8)}}
{\sqrt{(L-1) \cos{(|\mathcal{B}|/4)}}}, 
\\ \nonumber
\mathcal{I}&=\left(\frac{|\mathcal{B}|}{8}\right)^{L+1}\frac{2^{7/2}}{\sqrt{\pi L}}\left[1+\mathcal{O}\left(\left(\frac{|\mathcal{B}|}{8}\right)^2+ \frac{1}{L}\right)\right].
\end{align}
By comparing Eqs.~\ref{eq:BBInfidelity}.~\ref{eq:OBInfidelity} in Fig.~\ref{XGateBandwidth}, it can be seen that for any target $\mathcal{B}$ and sequence length $L$, the composite quantum gate $\text{OB}n=(\cdot,0,\text{M}_{2n+1,\mathcal{I}}(y),\cdot)$ has a worst-case infidelity that improves on $\text{BB}n$ by an exponential factor $\mathcal{O}(2^{1-L})$.
}
\begin{figure}
\includegraphics[width=\columnwidth]{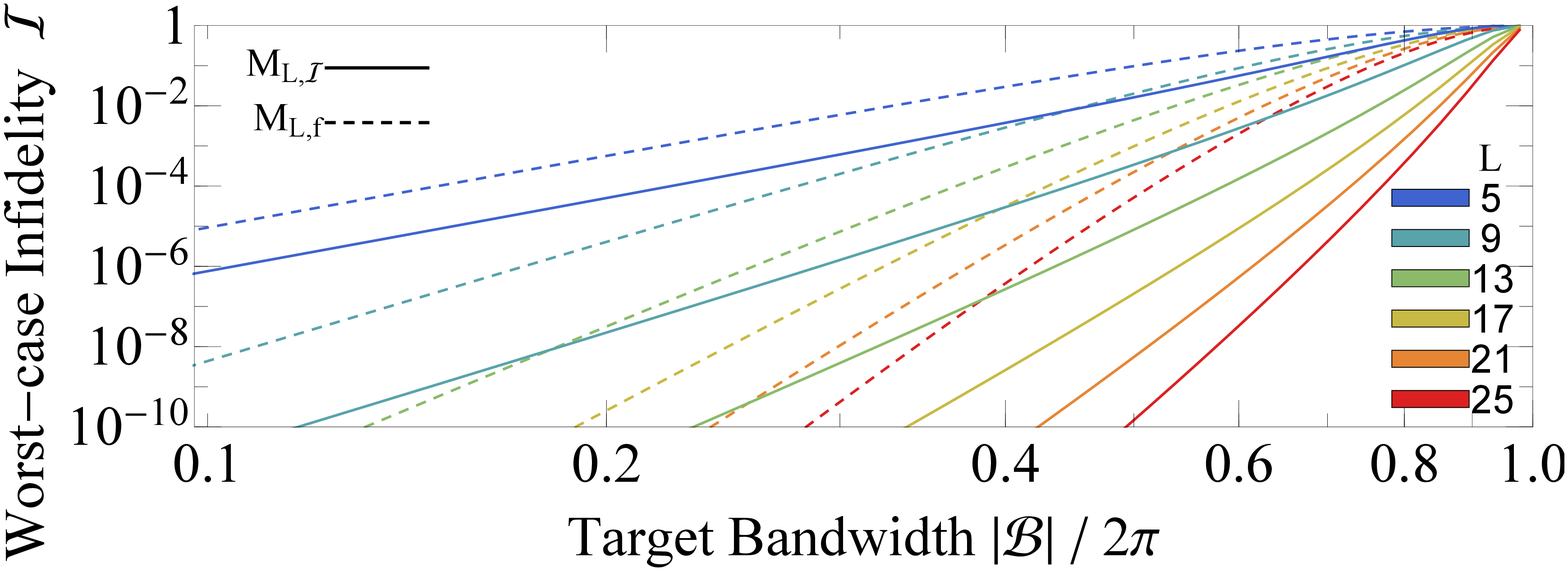}
\resizebox{\columnwidth}{!}{\begin{tabular}{c | c | c}
$\mathcal I$ & $L$ & $\vec{\phi}$ where $\phi_k=\phi_{L-k+1}$\\
\hline
Eq.~\ref{Eq:N5polynomial} & $5$ & $2(\tan^{-1}\tfrac{-t_2}{1+2x_1},\tan^{-1}t_2,0,...), t_2 = \sqrt{\tfrac{8 x_1^3 + 8 x_1^2 -1}{1+2 x_1}}$ \\
$10^{-2}$ & $9$ & $(2.987, 5.166, 4.021, 1.678, 2.815,...)$ \\
$10^{-4}$ & $9$ & $(2.889, 5.334, 4.042, 1.490, 2.926,...)$ \\
$10^{-6}$ & $9$ & $(2.844, 5.381, 4.034, 1.414, 2.976,...)$ \\
$10^{-2}$ & $13$ & $(2.390, 0.771, 2.791, 2.824, 2.115, 4.573, 4.888,...)$ \\
$10^{-4}$ & $13$ & $(2.233, 0.455, 2.853, 2.862, 1.838, 4.558, 5.041,...)$ \\
$10^{-6}$ & $13$ & $(2.159, 0.314, 2.874, 2.877, 1.677, 4.495, 5.092,...)$
\end{tabular}}
\caption{\label{XGateBandwidth}(color) Worst-case infidelity $\mathcal{I}$ of NOT gates $\text{OB}n=(\cdot,0,M_{2n+1,\mathcal{I}}(\sin{(\theta/2)}),\cdot)$ (solid, Eq.~\ref{eq:OBInfidelity}) optimized for target bandwidth $\theta\in\mathcal{B}$ compared to flatness-optimized NOT gate $\text{BB}n=(\cdot,0,M_{{2n+1},\text{f}}(\sin{(\theta/2)}),\cdot)$ (dashed, Eq.~\ref{eq:BBInfidelity}), plotted for $L=2n+1=5,9,...,25$ (from top). Observe that $\mathcal{I}$ for OB$n$ is exponentially smaller by factor $\approx 4^{n}$ than BB$n$. Alternatively, an OB$n$ gate can approximate NOT with infidelity at most $\mathcal{I}$ over a much wider bandwidth than BB$n$. The table provides examples of $\vec\phi$ for OB$n$ rounded to $3$ decimal places.}
\end{figure}
\ghl{
In contrast to the BB$n$ sequences that are fixed for each $n$, OB$n$ allows for an optimal design trade-off between bandwidth $\mathcal{B}$ and infidelity $\mathcal{I}$. As seen in Fig.~\ref{DCPlot}, this occurs by introducing \emph{equiripples} of equal amplitude bounded by $\mathcal{I}$, similar to the DC$_{L,\mathcal{I}}$ polynomials for population inversion gates. Thus, given the same performance targets, an extremely short OB$n$ gate can perform just as well as a significantly longer BB$n$ gate. In other words, maximizing the correction order only improves the achieved bandwidth indirectly, leading to a poor trade-off between $\mathcal{I}$ and $\mathcal{B}$, whereas better results are naturally achieved by optimizing for polynomials that directly solve Eq.~\ref{Eq:NOTBandwidthOptimization} by minimizing infidelity over a target bandwidth. 
}

\subsection{Composite quantum gates with sub-wavelength spatial selectivity}
\label{Sec:SpatialGates}
\ghl{
Narrowband compensated gates maximize the bandwidth $\mathcal{B}$ over which the fidelity response function with respect to identity $\hat{\mathbb{1}}$ is close to $1$, except at a single point $\theta$ where an arbitrary target rotation $\hat{R}_{0}(\chi)$ is applied. Although the direct approach is computing new polynomials $(A,\cdot,C,\cdot)$ that satisfy these properties, we can reuse the polynomials $M_{L,\mathcal{I}}$ from Sec.~\ref{Sec:OBGates} by making certain assumptions on the physical system. In the following, we also assume that $|\chi|\le \pi$.

Consider a Gaussian beam of fixed width $\lambda$. As a function of position $r$, this beam has a spatially-varying Rabi frequency $\Omega(r)=\Omega_0 e^{-r^2/2 \lambda^2}$. Thus when applied for time $t_0$, a primitive gate $\hat{R}_\phi(\theta(r))$ that also varies as a function of position is generated, where $\theta(r)=\theta_0 e^{-r^2/2 \lambda^2}$ and $\theta_0=\Omega_0t_0$. At $r=0$, one can choose $t_0,\phi$ such that the target rotation $\chi=\theta_0$ is implemented, and due to exponential decay of the Gaussian beam, moving away from the beam center approximates the identity gate with infidelity $\mathcal{I}(r)=\sin^2{(\frac{\theta(r)}{2})}$. Thus at distance $r/\lambda\ge d/\lambda  = \bar{\mathcal{B}}_1= \log^{1/2}{\frac{\pi^2}{4\mathcal{I}}}$ from the beam center, the worst-case infidelity is $\mathcal{I}$. As the minimum possible beam width $\lambda$ is the wavelength of light, selective addressing below the diffraction limit appears impossible. However, even this can be overcome with a carefully designed composite quantum gate.

Narrowband composite gates of length $L$ applicable to this scenario have been widely studied. For instance, ~\cite{Wimperis1994,Merrill2014} report beam width reductions by factor $\frac{d}{\lambda}\approx 0.7 \bar{\mathcal{B}}_1$~\cite{Wimperis1994,Merrill2014}. Further reduction is possible with longer composite gates~\cite{Low2014}, but with poor scaling $\frac{d}{\lambda}=\mathcal{O}(L^{-{1/4}})$. 

\tjyoder{A better narrowband composite gate results from} using the broadband identity gate $\text{ID}=(\text{M}_{L,\mathcal{I}}(x),0,\cdot,\cdot)$ designed from the $\text{M}_{L,\mathcal{I}}$ polynomial in presented in Sec.~\ref{Sec:OBGates}. \tjyoder{Then,} the fidelity response function with respect to identity is $F_{0}(\theta)=\text{M}_{L,\mathcal{I}}^2(x)$, which, as we now show, corresponds to a quadratic improvement \tjyoder{of} $\frac{d}{\lambda}=\mathcal{O}(L^{-{1/2}})$.

Let us compose $\text{ID}$ with the Gaussian beam to produce the spatially-varying quantum response function
\begin{align}
\hat{U}_{\text{space}}(r)=\text{ID}(\theta_0 e^{-r^2/2 \lambda^2})=\text{ID}(\theta_0)+\mathcal{O}(r^2),
\end{align}
for some choice $|\theta_0|\le \pi$. Note that $\hat{U}_{\text{space}}(r)$ is stable with respect to beam-pointing errors in $r$ due to the vanishing first derivative. The degree of spatial selectivity is computed from the bandwidth in Eq.~\ref{eq:OBInfidelity} by substituting $|\mathcal{B}|=2\theta_0e^{-\bar{\mathcal{B}}_{\text{space}}^2/2\lambda^2}$ and solving for $r/\lambda$. Thus, identity is implemented with infidelity at most $\mathcal{I}$ at all $r\ge d\ge \lambda\bar{\mathcal{B}}_{\text{space}}$ as seen in Fig.~\ref{Fig:NarrowbandPlot}, where to leading order $\mathcal{O}(L^{-1/2})$,
\begin{align}
\frac{\bar{\mathcal{B}}_{\text{space}}}{\lambda}=2\sqrt{\tfrac{\log{(1/\mathcal{I})}+\frac{1}{2}\log{(2^{7}/(L\pi))}}{L+1}-\ln{\tfrac{4}{\theta_0}}}.
\end{align}
Meanwhile at $r=0$, we obtain the gate \begin{align}
\hat{U}_{\text{space}}(0)=R_{\gamma}\left(2\cos^{-1}\left({\text{M}_{L,\mathcal{I}}\left(\cos{(\theta_0/2)}\right)}\right)\right).
\end{align}
where $\gamma=\text{Arg}[C(\sin{\frac{\theta_0}{2}})+i D(\sin{\frac{\theta_0}{2}})]$. The desired rotation $R_0(\chi)$ is thus obtained by choosing $\theta_0$ such that $\cos{(\chi/2)}=\text{M}_{L,\mathcal{I}}(\cos{(\theta_0/2)})$ and rotating all phases $\phi_k\leftarrow\phi_k+\gamma$, which follows from $e^{-i\frac{\gamma}{2}\sigma_z}\hat{U}_{\text{space}}(0)e^{i\frac{\gamma}{2}\sigma_z}=R_{0}(\chi)$.
}

\begin{figure}
\includegraphics[width=1.0\columnwidth]{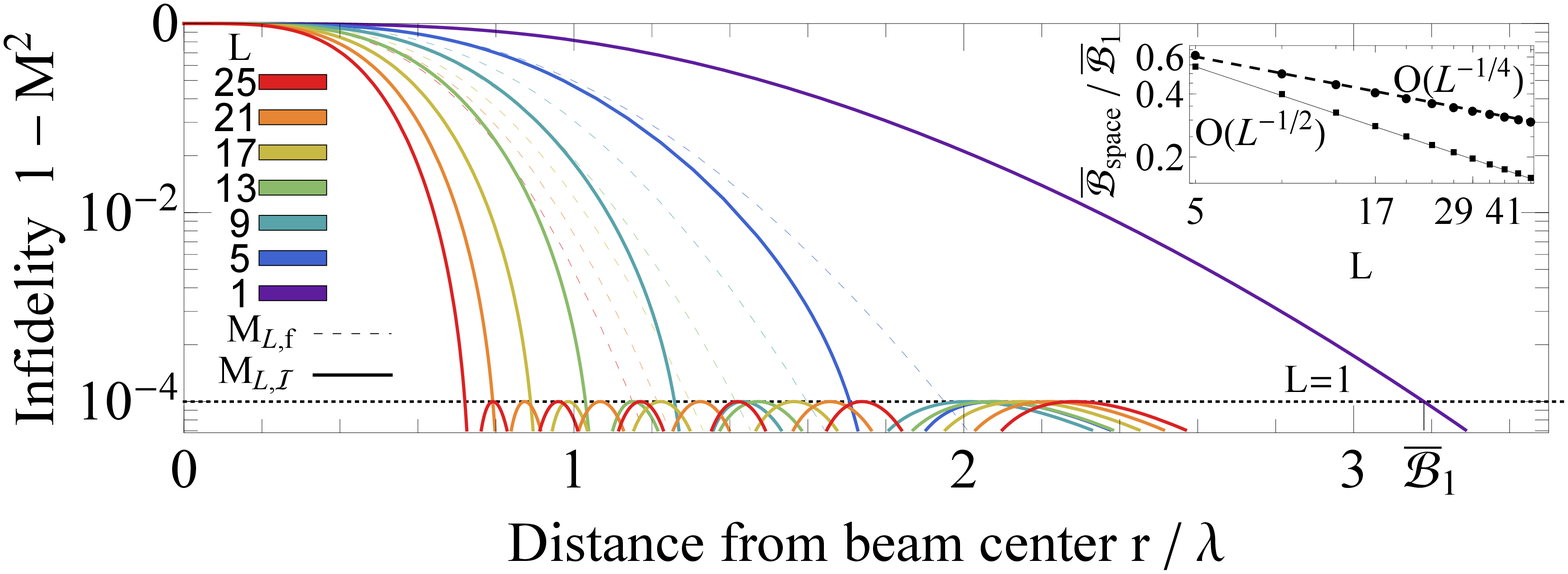}
\resizebox{\columnwidth}{!}{\begin{tabular}{c | c | c}
$\mathcal I$ & $L$ & $\vec{\phi}$ where $\phi_k=\phi_{L-k+1}$\\
\hline
Eq.~\ref{Eq:N5polynomial} & $5$ & $2(0, \tan^{-1}t_1,\tan^{-1}\tfrac{t_1}{1-4x_1^2},...), t_1 = \sqrt{8 x_1^3 + 8 x_1^2 -1}$ \\
$10^{-2}$ & $9$ & $(0, 0.772, 4.357, 2.827, 3.886,...)$ \\
$10^{-4}$ & $9$ & $(0, 1.087, 4.501, 2.707, 3.961,...)$ \\
$10^{-6}$ & $9$ & $(0, 1.235, 4.601, 2.695, 4.029,...)$ \\
$10^{-2}$ & $13$ & $(0, 1.450, 3.683, 2.501, 3.220, 5.577, 5.728,...)$ \\
$10^{-4}$ & $13$ & $(0, 1.872, 4.326, 2.844, 3.602, 6.271, 0.257,...)$ \\
$10^{-6}$ & $13$ & $(0, 2.077, 4.616, 2.978, 3.742, 0.250, 0.578,...)$
\end{tabular}}
\caption{\label{Fig:NarrowbandPlot}
(color)
Infidelity of spatially selective composite gates $(M_{L,10^{-4}}(\cos{(\frac{\theta_0}{2} e^{- r^2/\lambda^2})}),0,\cdot,\cdot)$ plotted for $\theta_0=\pi$ and $L=1,...,25$~(solid, from right). The effective beam width $\bar{\mathcal{B}}_{\text{space}}=\mathcal{O}(L^{-1/2})$(inset) beyond which the identity gate is well-approximated is dramatically reduced over that of a single gate $\bar{\mathcal{B}}_1$. By varying $\theta_0$, arbitrary unitary gates can be applied at $r=0$ with high beam-pointing stability. Poorer scaling $\bar{\mathcal{B}}_{\text{space}}=\mathcal{O}(L^{-1/4})$ results from using the flat $(M_{L,\text{f}},0,\cdot,\cdot)$ (dashed). The table provides examples of $\vec{\phi}$ to $3$ decimal places. 
}
\end{figure}
\ghl{
The optimality of these results follows from the construction of $\text{M}_{L,\mathcal{I}}$ as optimal bandwidth polynomials. In particular, using the flat polynomial $\text{M}_{L,\text{f}}(x)$ leads to the scaling $\bar{\mathcal{B}}_{\text{space}}=\mathcal{O}(L^{-1/4})$ found in prior art and Fig.~\ref{Fig:NarrowbandPlot}(inset). 
}
\section{Conclusion}
\label{Sec:Conclusion}
We have presented and applied a methodology, analogous to the Shinnar-LeRoux algorithm but with different controls, for the systematic design of \ghl{resonant equiangular} composite quantum gates of length $L$ on a single spin. In particular, we show that all steps are efficient with time complexity $\mathcal{O}(\text{poly}(L))$, and provide an extremely rigorous characterization of achievable quantum response functions. Moreover, the elegant and practical connection made with discrete-time signal processing allows us to inherit and adapt many existing algorithms and polynomials used in the design of \emph{classical} response functions for this \emph{quantum} problem. Much potential remains untapped there, and interdisciplinary exchange could spur the discovery of further connections, leading to the development of  previously intractable applications. Indeed, this relationship has already proven fruitful in surprising directions, such as recent work furnishing optimal algorithms for important problems such as Hamiltonian simulation~\cite{Low2016HamSim,Low2016HamSimQubitization} on a quantum computer.

In fact, our work bridges discrete-time signal processing and quantum query algorithms for evaluating symmetric boolean functions. The idea here is that the $\text{SU}(2)$ space of a single-qubit, as we study here, is isomorphic to the $\text{SU}(2)$ subspace spanned by a uniform superposition of marked states and a uniform superposition of unmarked state in such a query problem. Query algorithms can be built to calculate a boolean function $f:\{0,1\}^n\rightarrow\{0,1\}$ that depends only on the number of marked states (i.e. $f(x)=\tilde f(|x|)$ for some $\tilde f:\{0,1,\dots,n\}\rightarrow\{0,1\}$) and do so with a Grover-type algorithm of partial reflections (e.g. \cite{grover1996fast,brassard1998quantum,Hoyer2000arbitrary,grover2005fixed}). Thus, the same methods introduced here also give a way to determine how many reflections (analogous to our $L$) and what reflections (analogous to our $\phi_j$) are required to compute any particular symmetric boolean function, achieving the known lower bounds for this problem, which (not) coincidentally are also derived using polynomials \cite{beals2001quantum}. As examples of this correspondence, $DC_{L,\mathcal{I}}$ is an optimal solution for OR \cite{Yoder2014}, $\text{M}_{L,\mathcal{I}}$ is optimal for Majority.

Various thought provoking extensions are also motivated. The set of achievable quantum response functions is changed by introducing elements such as additional (possibly continuous) control parameters, disturbances, coupled spins~\cite{tomita2010multi,Merrill2014,Ivanov2015Composite}, or open systems~\cite{Khodjasteh2010arbitrarily,Soare2014}. These all enable their own unique applications, but also appear difficult to solve somehow systematically and intuitively. Our success in the case of composite gates contributes supporting evidence that a useful characterization as well as efficient methods for these more complex design problems could exist.

\section{Acknowledgements}
G.H.Low acknowledges funding by NSF RQCC Project No.1111337 and NRO. T.J.Yoder acknowledges funding by NDSEG. We thank Alan Oppenheim and Tom Baran for inspiring discussions, and connections made possible by their 6.341x open online MITx course. We thank Yuan Su for useful comments on the paper.

\bibliography{Methodology180201}

\end{document}